\documentclass{article}
\usepackage{latexsym}
\usepackage{amssymb,amsmath, amsthm, amscd}
\usepackage{setspace}
\usepackage{natbib}
\usepackage{breakcites}
\usepackage{authblk}
\allowdisplaybreaks

\usepackage{tikz}

\newtheorem{prop}{Proposition}

\newtheorem*{sym*}{SYM$^*$}
\newtheorem*{sym**}{SYM$^{**}$}
\newtheorem*{subgroup}{Subgroup}
\newtheorem*{sym+}{SYM$^+$}

\theoremstyle{definition}

\newtheorem{example1}{Example}

\theoremstyle{remark}

\tikzset{node distance=1.6cm, auto}

\newcommand\xqed[1]{%
  \leavevmode\unskip\penalty9999 \hbox{}\nobreak\hfill
  \quad\hbox{#1}}
\newcommand\corner{\xqed{$\lrcorner$}}

\newenvironment{example}[0]{\begin{example1}}{\corner\end{example1}}

\title{On Automorphism Criteria for Comparing Amounts of Mathematical Structure}
\author[1]{Thomas William Barrett\thanks{tbarrett@philosophy.ucsb.edu}}
\author[2]{JB Manchak\thanks{jmanchak@uci.edu}}
\author[2]{James Owen Weatherall\thanks{weatherj@uci.edu}}
\affil[1]{Department of Philosophy, University of California, Santa Barbara}
\affil[2]{Department of Logic and Philosophy of Science, University of California, Irvine}
\date{}
\begin{document}
\maketitle

\begin{abstract}
\cite{wilhelm2021} has recently defended a criterion for comparing structure of mathematical objects, which he calls Subgroup. He argues that Subgroup is better than SYM$^*$, another widely adopted criterion. We argue that this is mistaken; Subgroup is strictly worse than SYM$^*$.  We then formulate a new criterion that improves on both SYM$^*$ and Subgroup, answering Wilhelm's criticisms of SYM$^*$ along the way. We conclude by arguing that no criterion that looks only to the automorphisms of mathematical objects to compare their structure can be fully satisfactory.
\end{abstract}

\section{Introduction}

There is a long tradition in the philosophy of physics of arguments that one theory, or formulation of a theory, is superior to another, empirically equivalent theory, on grounds of structural parsimony.  The idea is that if two theories have the same empirical content, but one theory's models have less structure than the other theory's models, one should infer that the first theory attributes less structure to the world---and therefore should be preferred.  Over the past decade, much effort has been devoted to making the comparisons of ``amount of structure'' involved in such arguments precise.  

The most common type of criterion is based on the idea that one can compare amounts of structure by looking to the symmetries, or automorphisms, of the mathematical objects in question. If a mathematical object has more automorphisms, then it intuitively should have less structure that these automorphisms are required to preserve.  The amount of structure that a mathematical object has is, in some sense, inversely proportional to the size of the object's automorphism group. \citet[p.~36]{earman1989} puts this basic idea as follows: ``As the [\ldots] structure becomes richer, the symmetries become narrower.'' 

To have a precise criterion, one needs to clarify the sense in which one object may have `fewer' automorphisms than another. \cite{swansonhalvorsonunpublished} and \cite{barrett2014,barrett2014b} have proposed the following.

\begin{sym*}
A mathematical object $X$ has at least as much structure as a mathematical object $Y$ if (and only if) $\text{Aut}(X)\subseteq\text{Aut}(Y)$.
\end{sym*}

The condition $\text{Aut}(X)\subseteq\text{Aut}(Y)$, i.e.~that the automorphism group of $X$ $\text{Aut}(X)$ is a subset of that the automorphism group of $Y$ $\text{Aut}(Y)$, is one way to make precise the idea that $\text{Aut}(X)$ is `not larger than' $\text{Aut}(Y)$.\footnote{We here follow the statement of SYM$^*$ given in \cite{wilhelm2021}. \cite{barrett2014, barrett2014b} states SYM$^*$ using the relation ``more structure than'', rather than Wilhelm's ``at least as much structure as''.} SYM$^*$ makes intuitive verdicts in many easy cases of structural comparison.  Moreover, while it is not always explicitly mentioned, SYM$^*$ is the standard criterion in the literature. 

\cite{wilhelm2021} has recently argued that a different criterion, which he calls Subgroup, is superior to SYM$^*$.\footnote{Wilhelm introduces two criteria, Subgroup and Subgroup$_2$, where Subgroup$_2$ is strictly more liberal.  We focus on Subgroup in what follows because our intent is to argue that Subgroup is already too liberal. We return to Subgroup$_2$ in the conclusion.}  Subgroup, too, has been used in the literature.  For instance, Subgroup arises as an application of a category-theoretic criterion of structure comparison used by \citet{weatherallNG,weatherallgauge}, \citet{rosenstock2019}, and \citet{bradleyweatherall2020}, among others, based on the property-structure-stuff framework of \citet{baezsps}, to the case of categories consisting of single objects and their automorphisms.  Likewise, \citet[p. 823]{barrett2014} considers---and rejects---a very similar condition, which he calls SYM$^{**}$.\footnote{See note \ref{sym**}.}  But Wilhelm presents the only sustained defense of Subgroup in the literature, and so we focus on his treatment.

We have a few aims in this paper, all of which dovetail off of Wilhelm's discussion. First, we want to show that while SYM$^*$ does have shortcomings, Subgroup is worse.  As we will show, SYM$^*$ is weak because it cannot rule on pairs of objects with different underlying sets; Subgroup is weak(er) because it gets clear cases wrong.  Second, we will explain \emph{why} SYM$^*$ succeeds in the cases to which it applies: it captures an important sense of structural comparison, which we call the ``implicit definability conception'' (IDC).  Subgroup does not have any relationship to the IDC.  

Nonetheless, Wilhelm's arguments are suggestive. Drawing on Wilhelm's intuition, we propose another criterion that improves on both Subgroup and SYM$^*$.  This criterion is, in a sense we make precise, the best one can do using only automorphisms.  Even so, we will argue, it is not good enough. We will next argue that \textit{no} criterion for comparing amounts of structure that looks solely to automorphisms can be adequate. In brief, automorphisms alone do not encode all of the relevant facts about the structure of a mathematical object. We conclude with a brief discussion of how to think about comparisons of structure in light of the foregoing.

\section{SYM$^*$, Subgroup, and Their Problems}

Wilhelm's main concern with SYM$^*$ is a problem we will call \textbf{sensitivity}.\footnote{This problem is gestured at in the discussion of a criterion called SYM$^{**}$ by \cite{barrett2014}, and it is mentioned explicitly by \citet[p.~3]{barrett2014b}. It is also discussed by \cite{barrett2018a}.} In brief, SYM$^*$ is too sensitive to the underlying sets of the objects being compared.  To make this point, Wilhelm uses the example of two isomorphic groups that have different underlying sets. SYM$^*$ says that these two groups have `incomparable' amounts of structure, in the sense that neither has more nor less structure than the other. This is because no automorphism of the first group is also an automorphism of the second group. The two groups are isomorphic, however, and therefore should have \textit{the same} structure. \cite[p.~6361]{wilhelm2021} remarks that ``structural comparisons should imply that if two mathematical objects are isomorphic, then those objects have the same amount of structure. SYM$^*$ violates this condition. And that is a reason to reject it.'' 



We agree with Wilhelm in this case.  One can also find similar examples where one object should have \emph{more} structure than another, such as a topological space $(A,\tau)$ and a set $B$, where $A$ does not equal $B$.  Intuitively, one might think that a set with topology should have more structure than a bare set, even if the sets are not the same.  But again, SYM$^*$ does not rule on such cases.  Examples like these show that there is a sense in which SYM$^*$ is too \textit{strict} a criterion for comparing amounts of structure. There are pairs of objects $X$ and $Y$ such we want to say that $X$ has at least as much structure as $Y$, but SYM$^*$ does not make this verdict. 


It is to address sensitivity that Wilhelm proposes Subgroup.

\begin{subgroup}
A mathematical object $X$ has at least as much structure as a mathematical object $Y$ if (and only if) $\text{Aut}(X)$ is isomorphic to a subgroup of $\text{Aut}(Y)$.\footnote{\label{sym**} Subgroup differs from SYM$^{**}$, introduced and rejected by \citet{barrett2014}, in that SYM$^{**}$ says ``more'' where Subgroup says ``at least as much''.  Wilhelm's modification avoids the problem that some objects have more structure than themselves.}
\end{subgroup}

Despite the names, the important difference between Subgroup and SYM$^*$ is \emph{not} that Subgroup concerns the sub\emph{group} relation, whereas SYM$^*$ concerns the sub\emph{set} relation.  If $\text{Aut}(X)$ and $\text{Aut}(Y)$ are both automorphism groups, and $\text{Aut}(X)$ is a subset of $\text{Aut}(Y)$, then $\text{Aut}(X)$ is also a subgroup of $\text{Aut}(Y)$, since the two groups are groups of automorphisms on the same underling domain, and thereby have the same identity element and multiplication rule (composition).  The key difference is that SYM$^*$ compares group structure relative to a preferred embedding of $\text{Aut}(X)$ into $\text{Aut}(Y)$ (the identity map) determined by the fact that both are automorphisms on the same set. Subgroup, meanwhile, allows one to compare objects relative to \emph{any} injective group homomorphism between their automorphism groups.  No particular relationship between $X$ and and $Y$ is required or respected.  Thus we see a sense in which Subgroup is strictly `more liberal' than SYM$^*$.\footnote{Note that there are examples where SYM$^*$ and Subgroup agree that $X$ has at least as much structure as $Y$, but Subgroup \emph{also} rules that $Y$ has more structure than $X$ and SYM$^*$ does not.}

Subgroup does address sensitivity.  For instance, in Wilhelm's example of two isomorphic groups with different underlying sets, it is trivial to verify that they have the same amount of structure according to Subgroup. Likewise, if a topological space $(A, \tau)$ is compared to a set $B$ with the same cardinality as $A$, then Subgroup again makes the correct verdict: $(A,\tau)$ has at least as much structure as $B$. 

Wilhelm offers a few other reasons to prefer Subgroup to SYM$^*$.  One is that automorphism groups are groups, not sets, so the `subset relation' that SYM$^*$ employs is the wrong relation.  Second, Wilhelm argues that Subgroup is a ``strict generalization of SYM$^*$'' (p.~6365) and that this gives us reason to prefer it. Wilhelm explains why this is supposed to be a mark in favor of Subgroup as follows:
\begin{quote}
Subgroup expands the range of objects whose structures can be compared. So it supports more of the structural comparisons that mathematicians, physicists, and philosophers make. This, in fact, is the main reason why I prefer Subgroup [\ldots] to SYM$^*$. (p. 6365)
\end{quote}
Finally, Wilhelm claims that it makes intuitive verdicts. It makes many of the same verdicts that SYM$^*$ did.  It also deals well with cases like the isomorphic group case we have mentioned. 

As we have seen, Wilhelm's claim that Subgroup is superior because it explicitly involves the group structure of automorphism groups is a red herring. Both criteria involve group structure; they differ in what embeddings they allow.  Moreover, while it is true that Subgroup is a ``strict generalization'' of SYM$^*$, we contend that that by itself is not a reason to accept Subgroup.  A satisfactory generalization of SYM$^*$ needs to make sensible verdicts in the cases where it differs from SYM$^*$.  But Subgroup does not do this.  Consider, for example, the following verdicts that Subgroup makes.

\begin{itemize}
\item The group $\mathbb{Z}_5$ (automorphism group $\mathbb{Z}_4$) vs any set with cardinality 2.  Since the automorphism group of the latter can be properly embedded in the automorphism group of the former, according to Subgroup the set has at least as much structure as the group.
\item The vector space $\mathbb{R}^2$ vs the group $\mathbb{Z}$ (automorphism group $\mathbb{Z}_2$). $\mathbb{Z}$ has at least as much structure as $\mathbb{R}^2$ according to Subgroup, another unintuitive verdict, especially since the vector space has underlying group structure \textit{and} additional vector space structure on top of that.
\item The vector space $\mathbb{R}$ vs the vector space $\mathbb{R}^2$.  $\mathbb{R}$ has at least as much structure as $\mathbb{R}^2$ according to Subgroup, even though $\mathbb{R}^2$ contains many subspaces isomorphic to $\mathbb{R}$ and further structure relating those subspaces.  
\end{itemize}

We need not multiply examples; the ones given should suffice to make one nervous about Subgroup. In what follows, we will diagnose what is going wrong with that criterion.

\section{The Implicit Definability Conception}

We begin by discussing why SYM$^*$ works in the cases where it makes judgments. This argument has been before: see \cite{barrett2018a} and the references therein. The basic idea is that there is a sense in which $X$ has at least as much structure as $Y$ according to SYM$^*$ just in case $X$ actually has all of the structures that $Y$ has.  We can make this idea precise by considering some simple facts about definability in first-order logic and model theory.\footnote{See \cite{hodges2008} for further details.} 

A \textbf{signature} $\Sigma$ is a set of predicate symbols. (Our results generalize to the case of function and constant symbols as well.) The $\Sigma$-terms, $\Sigma$-formulas, and $\Sigma$-sentences are recursively defined in the standard way. A \textbf{$\mathbf{\Sigma}$-structure} $X$ is a nonempty set in which the symbols of $\Sigma$ have been interpreted. One recursively defines when a sequence of elements $a_1,\ldots, a_n\in X$ \textbf{satisfy} a $\Sigma$-formula $\phi(x_1,\ldots, x_n)$ in a $\Sigma$-structure $X$, written $X\vDash\phi[a_1,\ldots, a_n]$. We will use the notation $\phi^X$ to denote the set of tuples from the $\Sigma$-structure $X$ that satisfy a $\Sigma$-formula $\phi$. A \textbf{$\mathbf{\Sigma}$-sentence} is a $\Sigma$-formula with no free variables. 
An \textbf{automorphism} of a $\Sigma$-structure $X$ is a bijection from $X$ to itself that preserves the extensions of all of the predicates in $\Sigma$.

The basic set-up that we will employ in order to discuss definability is the following:
\begin{itemize}
\item Let $\Sigma_1$ and $\Sigma_2$ be signatures. The elements of $\Sigma_1$ and $\Sigma_2$ represent the `basic structures' on the two objects that we will consider. These can be thought of as the structures that are explicitly appealed to in the notation we use to describe the objects. 

\item Let $X$ be a $\Sigma_1$-structure and $Y$ a $\Sigma_2$-structure. We will think of $X$ and $Y$ as the two objects whose structures will we be comparing. We temporarily assume that $X$ and $Y$ have the same underlying set.
\end{itemize}
We need to make precise what it means for $X$ to define all of the basic structures of $Y$. So let $p\in\Sigma_2$ be one of the basic structures on $Y$. There are two particularly natural ways to make precise what it means for $X$ to define $p$. We say that $X$ \textbf{explicitly defines $p^Y$} if there is a $\Sigma_1$-formula $\phi$ such that $\phi^X=p^Y$. And we say that $X$ \textbf{implicitly defines} some subset $I\subset X\times\ldots\times X$ (like the structure $p^Y$) if $h[I]=I$ for every automorphism $h$ of $X$. We will focus on implicit definition.\footnote{There are several varieties of implicit definability in the literature. The condition we consider here is one of the weaker ones. See \cite{winnie1986}, \cite{barrett2017}, and references therein for further details.}

Here is the intuition behind these two notions of definability. If $X$ explicitly defines the structure $p^Y$, then $p^Y$ can be `constructed from' the basic structures in $\Sigma_1$.  On the other hand, suppose that $X$ implicitly defines $p^Y$. When this is the case, one often says that the structure $p^Y$ is `invariant under' or `preserved by' the symmetries of $X$. It is common to infer from this that $X$ comes equipped with the structure $p^Y$.\footnote{See \cite{dasgupta2014} or \cite{barrett2017} and the references therein for elaboration.}  The relation between these two varieties of definability is already well known. If $X$ explicitly defines $p^Y$, then $X$ implicitly defines $p^Y$. But the converse does not hold.\footnote{The converse \emph{would} hold if we restricted attention to complete theories or if we strengthened our definition of implicit definability.}

We have the following simple result.

\begin{prop}\label{prop:idc}
The following are equivalent:
\begin{enumerate}
\item For every symbol $p\in\Sigma_2$, $X$ implicitly defines $p^Y$.
\item $\text{Aut}(X)\subset \text{Aut}(Y)$.
\end{enumerate}
\end{prop}
\begin{proof}
Immediate from definitions.
\end{proof}

Prop. 1 establishes that SYM$^*$ says that $X$ has at least as much structure as $Y$, when $X$ implicitly defines all the structures of $Y$.  This is the Implicit Definability Conception of structural comparison mentioned above.  We take this to be a natural (weak) understanding of what it means to compare amounts of structure.  It is also the only conception of structural comparison that has been explicitly articulated and defended \citep[e.g. in][]{barrett2017,barrett2018a}.

The analogue of Prop. 1 does not hold for Subgroup. Subgroup does not bear the same relationship to definability that SYM$^*$ does.

\begin{example}
Let $\Sigma_1=\emptyset$ and $\Sigma_2=\{p\}$ be signatures, with $p$ a unary predicate symbol. Consider the $\Sigma_1$-structure $A$ whose underlying set is $\mathbb{N}$, and the $\Sigma_2$-structure $B$ whose underlying set is $\mathbb{N}$ with $p^B=\{0\}$. SYM$^*$ says $A$ has less structure than $B$. All the automorphisms of $B$ are automorphisms of $A$ but not vice versa. 

Subgroup says that they have the same amount of structure. Clearly $\text{Aut}(B)$ is isomorphic to a subgroup of $\text{Aut}(A)$. But conversely, the fact that $A$ is isomorphic to the set $B-p^B$ implies that $\text{Aut}(A)$ is isomorphic to $\text{Aut}(B)$, since every automorphism of $B$ is determined by its action on $B-p^B$.
\end{example}

In this example, $A$ and $B$ have the same amount of structure according to Subgroup, but $A$ does not implicitly define all of the structure of $B$.  Indeed, we generate $B$ by adding structure to $A$. It is hard to imagine a coherent understanding of structure according to which $B$ does not have more structure than $A$.  At the very least Subgroup does not implement the IDC.





The example also shows that the automorphism group of an object (up to isomorphism) does not fully encode how much structure an object has.  This means that a satisfactory criterion for comparing amounts of structure will need to appeal to \textit{more} than merely an object's automorphism group up to isomorphism. In the next section, we consider what else one might need.

\section{A Different Generalization of SYM$^*$}

As we have observed, the principal difference between SYM$^*$ and Subgroup is that SYM$^*$ compares the automorphism groups of $X$ and $Y$ only relative to a particular embedding, generated by a specific relationship between $X$ and $Y$: namely, the identity map on their respective domains.  If that map generates a group homomorphism, then $\text{Aut}(X)$ will be a subgroup of $\text{Aut}(Y)$.  But it is only the group homomorphism (possibly) generated by this particular relationship between $X$ and $Y$ that matters. 

These remarks suggest a different way to generalize SYM$^*$ to solve sensitivity.  The proposal makes use of the following Proposition.  In what follows, for any $\Sigma$-structure $Z$, $\text{dom}(Z)$ will refer to the domain of $Z$, i.e., its underlying point set. As before, let $X$ be a $\Sigma_1$-structure and $Y$ a $\Sigma_2$-structure. 

\begin{prop}\label{prop:sym+}
Suppose $f:\text{dom}(Y)\rightarrow \text{dom}(X)$ is injective. Then the following are equivalent:
\begin{enumerate}
    \item there is a group homomorphism $F:\text{Aut}(X)\rightarrow\text{Aut}(Y)$ that commutes with $f$, in the sense that $s\circ f=f\circ Fs$ for every automorphism $s$ of $X$.
    \item $X$ implicitly defines $f[I]$ for every subset $I$ that $Y$ implicitly defines.
\end{enumerate}
Moreover, if an $F$ as described in 1. exists, it is unique.
\end{prop}

Note that condition 1 does not quite imply that $F$ is injective, and so $F$ does not establish that $\text{Aut}(X)$ is isomorphic to a subgroup of $\text{Aut}(Y)$.  But it does imply the following, which is `close' to $F$ being injective (and becomes `closer' the `closer' to surjective $f$ is): If $Fs=Fs'$ for automorphisms $s$ and $s'$ of $X$, then $s|_{f[Y]}=s'|_{f[Y]}$. This follows immediately from the fact that $f$ is injective. For suppose that $f(x)\in f[Y]$. We know that 
$$
s\circ f(x)=f\circ Fs(x)=f\circ Fs'(x)=s'\circ f(x)
$$
The first and third equalities follow from condition 1, while the second follows from our assumption that $Fs=Fs'$. Since $f(x)$ was an arbitrary element of $f[Y]$, we have that $s$ and $s'$ are equal when restricted to $f[Y]$.

\begin{proof}
Both directions of the proof follow simply from definitions.
\end{proof}

This proposition suggests a different criterion for comparing structure.

\begin{sym+}
A mathematical object $X$ has at least as much structure as a mathematical object $Y$, relative to an injective function $f:Y\rightarrow X$, if (and only if) there exists a group homomorphism  $F:\text{Aut}(X)\rightarrow \text{Aut}(Y)$ that commutes with $f$.
\end{sym+}

SYM$^+$, like Subgroup, is a strict generalization of SYM$^*$.  SYM$^*$ is just the special case of SYM$^+$ where $f$ is the identity on $Y$.  SYM$^+$ is also compatible with the IDC, as Prop.~2 shows. And SYM$^+$ also solves sensitivity, since we can compare structures with different domains using SYM$^+$.  But according to SYM$^+$, this can happen \emph{only} relative to a particular injective map.  

In many cases, the map $f$ is fixed by context.  For instance, isomorphic groups have the same amount of structure, according to SYM$^+$, relative to any map $f$ that realizes their isomorphism.  A topological space $(A,\tau)$ has at least as much structure as a set $B$ of the same cardinality, according to SYM$^+$, relative to any bijection between $A$ and $B$.  
Still, one might worry that structural comparisons should not require a map between the objects.  How much structure something has does not depend on maps to other things; thus, comparing the structure of two different things have should not depend on maps either.  But this intuition fails for criteria that compare automorphism groups.  

The reason is that an automorphism group is not just a group.  It is a particular kind of \emph{group representation}: specifically, a representation of a group as the automorphisms on a given object.  The injective map $f$ is important because it allows us to compare $\text{Aut}(X)$ and $\text{Aut}(Y)$ \emph{as automorphism groups}.  Results like Props. \ref{prop:idc} and \ref{prop:sym+} show that to capture information about implicit definability using automorphism groups, one must keep track of how the abstract group structure is represented as maps on the domain of the object.  In order to compare what structures are implicitly definable for different objects using their automorphism groups, one needs to know not only how the groups are related, but how their representations as automorphism groups are related. 

Another worry about SYM$^+$ is that, since it considers only \emph{injective} maps $f$,  it will never rule that an object $X$ has at least as much structure as an object $Y$ if $Y$ has greater cardinality than $X$.  This is because if $Y$ has more elements than $X$, there will be no injective maps from $Y$ to $X$.  This may seem puzzling, but we set this it aside until the end of the next section.  

We have emphasized how SYM$^+$ relates to SYM$^*$.  But it is also related to Subgroup.  Consider the following proposition.  

\begin{prop}\label{prop:subgroup}
Let $f:\text{dom}(Y)\rightarrow \text{dom}(X)$ be a bijection. Then the following are equivalent:
\begin{enumerate}
    \item there is an injective group homomorphism $F:\text{Aut}(X)\rightarrow\text{Aut}(Y)$ that commutes with $f$, in the sense that $f\circ s=Fs\circ f$ for every automorphism $s$ of $X$.
    \item $X$ implicitly defines $f[I]$ for every subset $I$ that $Y$ implicitly defines.
\end{enumerate}
Moreover, if an $F$ as described in 1. exists, it is unique.
\end{prop}

\begin{proof}
This follows immediately from Proposition 2 and the remark following it.
\end{proof}

Prop. \ref{prop:subgroup} tells us that SYM$^+$ looks like Subgroup in the special case where $X$ and $Y$ have equinumerous domains.  Relative to a bijection $f:Y\rightarrow X$, SYM$^+$ says $X$ has at least as much structure as $Y$ only if $\text{Aut}(X)$ is isomorphic to a subgroup of $\text{Aut}(Y)$.  Thus, Subgroup provides a necessary condition.  But it is not sufficient.  The subgroup of $\text{Aut}(Y)$ must be the image of a homomorphism from $\text{Aut}(X)$ that commutes with $f$.  Some bijection must be specified.  Different choices of bijection may produce different verdicts.  

\section{Triviality}

We have now seen that SYM$^+$ is an automorphism-based criterion of structural comparison that is compatible with the ICD and solves sensitivity.  But even so, SYM$^+$ is not fully adequate. The reason stems from another problem, which we will \textbf{triviality}.  Triviality is a problem for all  automorphism criteria of structure comparison.\footnote{Triviality has been discussed by \cite{barrett2018a}; see also \citet{WeatherallWNCE}, who raises a related triviality concern for category theoretic approaches.}  Wilhelm mentions this problem but does not fully appreciate it (p.~6365).  Here we present a new example and some new arguments that clarify why triviality matters.  

In brief, the problem is that all criteria that look (only) to automorphism groups make implausible verdicts about objects that have trivial automorphism groups, i.e., whose only automorphism is the identity map.  Examples of such objects are: any set with one element; any vector space with a fixed basis; the group $\mathbb{Z}_2$; any prime field; and so on. Such objects form a diverse collection; automorphism-based criteria struggle with that diversity. 

Consider SYM$^+$.\footnote{\cite{barrett2018a} offers an example to make the same point in connection with SYM$^*$.}  According to this criterion, given any two objects $X$ and $Y$, if $X$ has trivial automorphism group and cardinality as least as great as $Y$, then $X$ has at least as much structure as $Y$. This is true no matter what structure $X$ and $Y$ actually carry.  The problem is arguably even worse for Subgroup.  According to Subgroup, given any mathematical structure $Y$, \emph{every} object with trivial automorphism group has at least as much structure as $Y$.    

Wilhelm is aware of triviality. But he argue that it ``has some independent motivation'':
\begin{quote}
For suppose $X$ is a spacetime with a trivial automorphism group. Then only the identity transformation preserves all of the structure of $X$. Only the identity transformation leaves $X$ invariant. In other words, the spatiotemporal structure of $X$ is so rich and complicated that every other transformation fails to preserve it. Therefore, $X$'s structure is `maxed out'. $X$ is as structured as can be. \citep[p. 6365]{wilhelm2021}
\end{quote}
Something similar can be said for SYM$^+$.  According to SYM$^+$, if an object has trivial automorphism group, no other object can have more structure, relative to any injective map $f$.  SYM$^+$, too, says such objects are ``maxed out''.

This is the wrong verdict.  Consider the following example.

\begin{example}
It is well known that there exist spacetimes with trivial isometry groups. David Malament has sketched an elegant way to construct an example: start with Minkowski spacetime and then excise a compact region ``shaped like a giraffe'' from the manifold. Here, we present a precise variation of this idea. We restrict attention to the giraffe region itself, take its interior, and consider it as a spacetime in its own right. The resulting example is flat and, if the (radically idealized!) giraffe region is suitably chosen, it has an underlying manifold diffeomorphic to $\mathbb{R}^n$.

\begin{figure}[htp!] \centering
 \includegraphics[width=2in]{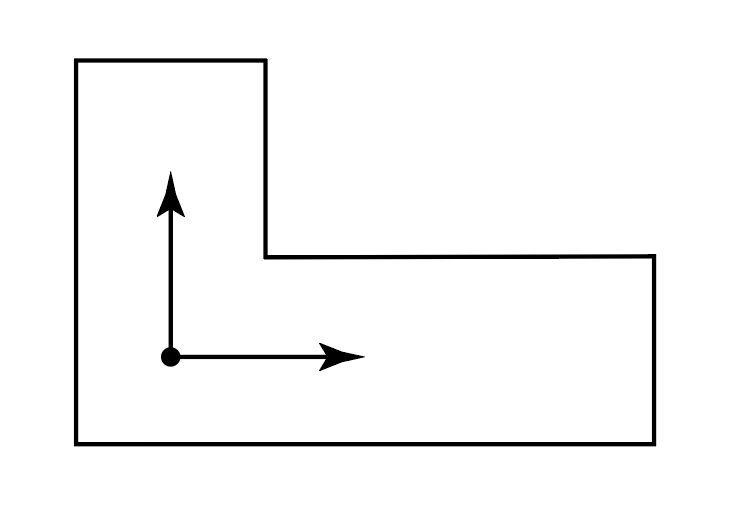} 
 \caption{The `giraffe' spacetime $(M, g_{ab})$ with vectors $t^a$ and $x^a$ at p.}.
\end{figure}

Let $(\mathbb{R}^2, g_{ab})$ be two-dimensional Minkowski spacetime where $g_{ab}=\nabla_at\nabla_bt-\nabla_ax\nabla_bx$. Let $M_1=\{(t,x) \in \mathbb{R}^2: 0<t<2, 0<x<3\}$. Let $M_2=\{(t,x) \in M_1: 1<t<2, 1<x<3\}$. Let $M=M_1 - M_2$ and consider the spacetime $(M, g_{ab})$. We see that since $M$ is a non-empty open star domain in $\mathbb{R}^2$, it is diffeomorphic to $\mathbb{R}^2$. Let $p=(1/2,1/2)$ and consider the vectors $t^a=(\partial/\partial t)^a$ and $x^a=(\partial/\partial x)^a$ at $p$ (see Figure 1). One can verify that any isometry $\varphi: M \rightarrow M$ is such that $\varphi(p)=p$, $\varphi_*(t^a)=t^a$, and $\varphi_*(x^a)=x^a$. From this it follows that $\varphi$ is the identity map given that isometries are rigid to first order at a given point (Geroch 1969).  
\end{example}


This example is striking because, while it has a trivial automorphism group, \emph{locally} its structure is hardly maxed out.  For instance, one could add an orientation field $\epsilon_{ab}$ to this spacetime.  That should generate an object with more structure, since metric+orientation is more structure than metric. But none of the criteria under consideration reflect that. One might also compare this example to a similar construction where one begins with a metric that admits fewer automorphisms.  For instance, if Newtonian spacetime has more structure than Minkowski spacetime, then a giraffe modeled on Newtonian spacetime should have more structure than a Minkowski giraffe.  None of the criteria we have discussed capture this, either.

The moral we wish to draw is that no automorphism criterion will fully capture all of the structural comparisons one might wish to make.  SYM$^+$ is useful for some kinds of cases, but not for others---much like SYM$^*$.  On the other hand, one can show a sense in which the giraffe intuitions just described may be made precise in a way that is compatible with the IDC.

Suppose again that $X$ is a $\Sigma_1$-structure and $Y$ a $\Sigma_2$-structure. If $f:\text{dom}(X)\rightarrow \text{dom}(Y)$ is an injective map, we let $f[X]$ be the $\Sigma_2$-structure obtained from $Y$ by `restricting' $Y$ to the image of $f$. Then we have the following.

\begin{prop}\label{prop:local}
Let $f:\text{dom}(X)\rightarrow \text{dom}(Y)$ be injective. Then the following are equivalent:
\begin{enumerate}
    \item there is a group homomorphism $F:\text{Aut}(X)\rightarrow\text{Aut}(f[X])$ that commutes with $f$, in the sense that $f\circ s=Fs\circ f$ for every automorphism $s$ of $X$.
    \item $X$ implicitly defines $f^{-1}[I]$ for every subset $I$ that $f[X]$ implicitly defines.
\end{enumerate}
Moreover, if an $F$ as described in 1. exists, it is unique.
\end{prop}

In this case, when condition 1 holds, $F$ is guaranteed to be injective, by a similar argument as given above.

\begin{proof}
The proof once again follows from definitions.
\end{proof}

This proposition is helpful in a few ways.  First, it captures a sense in which an object $X$ may have at least as much structure as `part' of another object $Y$ (In the proposition the `part' is represented by the structure $f[X]$.)  In the giraffe example, suitably chosen regions of two-dimensional Newtonian spacetime have at least as much structure as (regions) of the Minkowski giraffe.  This construction allows us to say how one can `add' structure to the Minkowski giraffe even though its automorphism group is trivial.

Prop. \ref{prop:local} is also useful in that it clarifies how an object $X$ may have more structure than an object $Y$ with greater cardinality, relative to some map $f$ (now going from $X$ to $Y$).  The proposition shows that $X$ may have more structure than the \emph{part} of $Y$ (given by $f[X]$) when $X$ defines all of the structure that $Y$ has on the image of $X$ under $f$. 


\section{Conclusion}

This paper has considered several automorphism-based criteria for comparing the amount of structure of mathematical objects.  They all faced difficulties.  SYM$^*$ suffers from sensitivity, whereas Subgroup makes non-sensical verdicts.   We went on to introduce a new criterion for structure comparison, SYM$^+$, which extends SYM$^*$, solves sensitivity, and conforms with the IDC.  We also showed how SYM$^+$ captures (some of) the intuition behind Subgroup.

Even so, we then argued that no automorphism-based criterion is completely satisfactory, because all of them suffer from triviality.  Triviality is not new, but we introduced an example that we feel highlights the difficulty.  The example shows that objects with trivial automorphism groups need not have ``maxed out'' structure, since in many cases, one can add additional structure to such objects.

We wish to conclude by tying up three loose ends.  First, observe that we did not propose a new criterion for structural comparison modeled on Prop. \ref{prop:local}, as SYM$^+$ was modeled on Prop. \ref{prop:sym+}.  One could do so.  But we would urge a different perspective.  Prop. \ref{prop:local} captures a fine-grained relationship whose interpretation depends on things like the map $f$ and the relationship between $f[X]$ and $Y$.  It would not be the right criterion in all cases.  We think the right moral to draw from the arguments in the foregoing is that (1) the IDC is the right way of thinking about structural comparisons; but (2) there is not a single criterion that implements that conception in all cases.  Instead, we would endorse a tool-box view, where context and careful consideration of the questions at hand guide which of a variety of precise criteria one uses to compare the structure of different objects.  

The second remark concerns Subgroup$_2$, the second criterion of structural comparison proposed by Wilhelm.  This criterion adds to Subgroup the possibility that one object has at least as much structure as another if the former's automorphism group can be generated from that of the latter as the limit of a one-parameter family of representations.  This proposal strikes us as an \emph{ad hoc} attempt to save the idea that Galilean spacetime has more structure than Minkowski spacetime, since the Galilean group can be generated as the limit of a one parameter family of representations of the Poincar\'e group.  But whatever its motivations might be, we observe that there is no relationship between group contraction and the IDC.  Thus, if group contraction is to be included in a criterion of structure comparison, some other view of what the criteria intend to capture must be articulated.

Finally, we return to a remark from the Introduction, that Subgroup arises as the specialization of another criterion of structure comparison, based on the \citet{baezsps} ``property-structure-stuff'' forgetful functor approach, to the case where one is comparing individual objects.  This approach, in general, has been adopted by several philosophers to compare, for instance, different formulations of Newtonian gravitation and electromagnetism.  But our arguments against Subgroup are also compelling objections to this category-theoretic approach.  We suggest that something is missing from the category theoretic approach, and that more work is needed on that subject.




\section*{Acknowledgments}
The authors are grateful to Hans Halvorson and David Malament for many discussions about topics related to this paper.  We are also grateful to Isaac Wilhelm for his comments on the manuscript.

\bibliographystyle{apalike}
\bibliography{masterbib}

\begin{thebibliography}{}

\bibitem[Baez et~al., 2006]{baezsps}
Baez, J., Bartels, T., Dolan, J., and Corfield, D. (2006).
\newblock Property, structure and stuff.
\newblock Available at
  http://math.ucr.edu/home/baez/qg-spring2004/discussion.html.

\bibitem[Barrett, 2015a]{barrett2014}
Barrett, T.~W. (2015a).
\newblock On the structure of classical mechanics.
\newblock {\em The British Journal for the Philosophy of Science},
  66(4):801--828.

\bibitem[Barrett, 2015b]{barrett2014b}
Barrett, T.~W. (2015b).
\newblock Spacetime structure.
\newblock {\em Studies in History and Philosophy of Science Part B: Studies in
  History and Philosophy of Modern Physics}, 51:37--43.

\bibitem[Barrett, 2018]{barrett2017}
Barrett, T.~W. (2018).
\newblock What do symmetries tell us about structure?
\newblock {\em {P}hilosophy of {S}cience}, 85:617--639.

\bibitem[Barrett, 2021]{barrett2018a}
Barrett, T.~W. (2021).
\newblock How to count structure.
\newblock {\em Forthcoming in {N}o\^us}.

\bibitem[Bradley and Weatherall, 2020]{bradleyweatherall2020}
Bradley, C. and Weatherall, J.~O. (2020).
\newblock On representational redundancy, surplus structure, and the hole
  argument.
\newblock {\em Forthcoming in Foundations of Physics}.

\bibitem[Dasgupta, 2016]{dasgupta2014}
Dasgupta, S. (2016).
\newblock Symmetry as an epistemic notion (twice over).
\newblock {\em The British Journal for the Philosophy of Science},
  67(3):837--878.

\bibitem[Earman, 1989]{earman1989}
Earman, J. (1989).
\newblock {\em World Enough and Spacetime: Absolute versus Relational Theories
  of Space and Time}.
\newblock MIT.

\bibitem[Hodges, 2008]{hodges2008}
Hodges, W. (2008).
\newblock {\em Model Theory}.
\newblock Cambridge University Press.

\bibitem[Rosenstock, 2019]{rosenstock2019}
Rosenstock, S. (2019).
\newblock {\em A Categorical Consideration of Physical Formalisms}.
\newblock PhD thesis, UC Irvine.

\bibitem[Swanson and Halvorson, 2012]{swansonhalvorsonunpublished}
Swanson, N. and Halvorson, H. (2012).
\newblock On {N}orth's `{T}he structure of physics'.
\newblock {\em Manuscript}.

\bibitem[Weatherall, 2016a]{weatherallNG}
Weatherall, J.~O. (2016a).
\newblock Are {N}ewtonian gravitation and geometrized {N}ewtonian gravitation
  theoretically equivalent?
\newblock {\em Erkenntnis}, 81(5):1073--1091.

\bibitem[Weatherall, 2016b]{weatherallgauge}
Weatherall, J.~O. (2016b).
\newblock Understanding gauge.
\newblock {\em {P}hilosophy of {S}cience}, 83(5):1039--1049.

\bibitem[Weatherall, 2021]{WeatherallWNCE}
Weatherall, J.~O. (2021).
\newblock Why not categorical equivalence?
\newblock In Madar{\'a}sz, J. and Sz{\'e}kely, G., editors, {\em Hajnal
  Andr{\'e}ka and Istv{\'a}n N{\'e}meti on Unity of Science: From Computing to
  Relativity Theory Through Algebraic Logic}, pages 427--451. Springer
  International Publishing, Cham.

\bibitem[Wilhelm, 2021]{wilhelm2021}
Wilhelm, I. (2021).
\newblock Comparing the structure of mathematical objects.
\newblock {\em Forthcoming in {S}ynthese}.

\bibitem[Winnie, 1986]{winnie1986}
Winnie, J. (1986).
\newblock Invariants and objectivity: A theory with applications to relativity
  and geometry.
\newblock In Colodny, R.~G., editor, {\em From Quarks to Quasars}, pages
  71--180. Pittsburgh: Pittsburgh University Press.

\end{thebibliography}
\end{document}